\documentclass[oneside]{amsart}
\usepackage{mathptmx}
\usepackage[T1]{fontenc}
\usepackage[latin9]{inputenc}
\usepackage{amstext}
\usepackage{amsthm}
\usepackage{amssymb}
\usepackage[dvipdfmx]{graphicx}
\usepackage{esint}
\usepackage[numbers]{natbib}

\makeatletter
\numberwithin{equation}{section}
\numberwithin{figure}{section}
\theoremstyle{plain}
\newtheorem{thm}{\protect\theoremname}
  \theoremstyle{plain}
  \newtheorem{lem}[thm]{\protect\lemmaname}
  \theoremstyle{plain}
  \newtheorem{prop}[thm]{\protect\propositionname}

\usepackage{url}

\usepackage{mathptmx}      

\@ifundefined{showcaptionsetup}{}{%
 \PassOptionsToPackage{caption=false}{subfig}}
\usepackage{subfig}
\makeatother

  \providecommand{\lemmaname}{Lemma}
  \providecommand{\propositionname}{Proposition}
\providecommand{\theoremname}{Theorem}

\begin{document}

\title{$\mathcal{R}_{0}$ fails to predict the outbreak potential in the
presence of natural-boosting immunity}

\thanks{Both authors have equally contributed to this paper.}

\author{Yukihiko Nakata}

\author{Ryosuke Omori}

\address[Y. Nakata]{Department of Mathematics, Shimane University, 1060 Nishikawatsu-cho,
Matsue, Japan}

\email{ynakata@riko.shimane-u.ac.jp}

\address[R. Omori]{Division of Bioinformatics, Research Center for Zoonosis Control,
Hokkaido University, Sapporo, Hokkaido, Japan\\
 JST, PRESTO, 4-1-8 Honcho, Kawaguchi, Saitama, 332\textendash 0012,
Japan}

\email{omori@czc.hokudai.ac.jp}
\begin{abstract}
Time varying susceptibility of host at individual level due to waning
and boosting immunity is known to induce rich long-term behavior of
disease transmission dynamics. Meanwhile, the impact of the time varying
heterogeneity of host susceptibility on the shot-term behavior of
epidemics is not well-studied, even though the large amount of the
available epidemiological data are the short-term epidemics. Here
we constructed a parsimonious mathematical model describing the short-term
transmission dynamics taking into account natural-boosting immunity
by reinfection, and obtained the explicit solution for our model.
We found that our system show ``the delayed epidemic'', the epidemic
takes off after negative slope of the epidemic curve at the initial
phase of epidemic, in addition to the common classification in the
standard SIR model, i.e., ``no epidemic'' as $\mathcal{R}_{0}\leq1$
or normal epidemic as $\mathcal{R}_{0}>1$. Employing the explicit
solution we derived the condition for each classification. 

\end{abstract}

\maketitle

\section{Introduction}

Modelling the transmission dynamics of infectious diseases and the
estimation of its model parameters are essential to understand the
transmission dynamics. Susceptible-infective-removed model, so-called
SIR model is known to be the simplest model to describe the transmission
dynamics \citep{Anderson:1992,Diekmann:2012}. The SIR model describes
transmission of pathogen from infective individuals to susceptible
individuals and removing infective individuals from the targeted host
population due to the establishment of immunity or death of host or
host immigration. Due to the wide variation in the natural history
of pathogen, many extended models from the basic SIR model have been
proposed so far.

An important extension is the time-evolution of susceptibility against
the infection with a pathogen. The basic SIR model describes that
the host immunity perfectly protects the host from reinfection over
time, then reinfection cannot occur forever. Meanwhile, reinfection
events are observed frequently among many infectious diseases, e.g.,
Coronavirus \citep{Isaacs:1983}, Respiratory syncytial virus \citep{Hall:1991},
Tuberculosis \citep{Verver:2005} and Hepatitis C virus \citep{vandeLaar:2009}.
One of considerable mechanisms of reinfection is waning immunity.
Decreased herd immunity by waning immunity of individuals induces
re-emergence of epidemic, and boosting immunity by re-vaccination
is required to control epidemics \citep{Barbarossa:2015}. Another
mechanism is imperfectness of immunity by an infection event. The
booster dose of vaccine is required to establish the high enough immunity
level to protect hosts from reinfection \citep{Siegrist:2008}, this
implies that the multiple exposures to the pathogen is required to
establish the high enough immunity level. Moreover, the enhancement
of susceptibility to reinfection is also observed among several infectious
diseases \citep{Wei:2011,Kohlmann:2004}.

Epidemic models incorporating variable susceptibility of recovered
individuals was formulated in the papers \citep{Kermack:1932,Kermack:1933}
by Kermack and McKendrick. However, the authors did not obtain a clear
biological conclusion \citep{Diekmann:1995,Inaba:2016}. In \citep{Inaba:2016,Inaba:2001}
the author performed stability analysis for the Kermack and McKendrick's
reinfection model formulated as a system of partial differential equations.
The existence and bifurcation of the endemic equilibrium is analyzed
in detail. Destabilization of the endemic equilibrium was shown to
be possible for epidemic models with waning immunity \citep{Diekmann:1982,Hethcote:2005,Nakata:2014}.
Previous modeling studies showed that waning and natural-boosting
immunity by exposure to the pathogens can trigger a counter-intuitive
effect of vaccination \citep{Lavine:2011}. It is suggested that waning
immunity in vaccinated hosts can trigger backward bifurcation of the
endemic equilibrium \citep{Arino:2003,Brauer:2004,Kribs:2000}. Estimating
the vaccine effectiveness is essential to control epidemics, however,
vaccine effectiveness reflects the complicated epidemiological dynamics
which is scaled by waning and natural-boosting immunity, e.g., boosting
and waning immunity can induce the periodic outbreak for the long-term
behavior \citep{Arinaminpathy:2012}.

Compared to the long-term behavior, the short-term behavior with waning
and boosting immunity is not well understood, although many field
data of the short-term epidemics have been analyzed using the model
without such waning and boosting immunity. As for short-term behavior,
the dynamics with constant immune protection rate against reinfection
has been studied so far while boosting and waning immunity change
the immune protection rate. In \citep{Katriel:2010} the author analyzed
transient dynamics of a reinfection epidemic model, ignoring the demographic
process in the model studied in \citep{Gomes:2004}. In the model,
reinfection of recovered individuals occurs, assuming that recovered
individuals have suitable susceptibility to the disease. It was shown
that the disease transmission dynamics qualitatively changes, when
the basic reproduction number crosses the reinfection threshold.

In this paper, we constructed a mathematical model taking into account
natural-boosting immunity. Since the time scale of waning immunity
is relatively longer than transmission dynamics, e.g., minimal annual
waning rate of immunity is $-2.9\%$ for rubella and $-1.6\%$ for
measles \citep{Kremer:2006}, compared to the infectious periods,
$11$ days for rubella \citep{Edmunds:2000} and $14$ days for measles
\citep{Anderson:1992}, we here focus on only boosting immunity. Since
boosting and waning immunity can induce periodic outbreak for the
long-term behavior \citep{Arinaminpathy:2012}, complicated epidemic
curve may be observed in the model for a short-term disease transmission
dynamics. We here obtain an explicit solution for the number of infective
individuals, consequently, we investigated how the short-term behavior
of the transmission dynamics is influenced by boosting immunity. The
shape of the short-term epidemic curve is analyzed in detail.

The paper is organized as follows. In Section 2 we formulate an epidemic
model, taking into account natural-boosting immunity, by a nonlinear
system of differential equations. The model includes the standard
SIR epidemic model and the reinfection epidemic model studied in \citep{Katriel:2010}
as special cases. In Section 3 we study the disease transmission dynamics
when the basic reproduction number, which is denoted by $\mathcal{R}_{0}$,
exceeds one. The number of the epidemic curve is shown to be one,
as is the case for the standard SIR epidemic model. In Section 3 we
consider disease transmission dynamics when $\mathcal{R}_{0}\leq1$.
Here we show that epidemic occurs even if $\mathcal{R}_{0}\leq1$,
due to the enhancement of susceptibility of recovered individuals.
We analyze the shape of the epidemic curve in detail. In Section 4,
the final size relation is derived from the explicit solutions in
the phase planes. In Section 6 we discuss our results for the future
works.


\section{An epidemic model with natural-boosting immunity}

First of all let us introduce the epidemic model studied in \citep{Katriel:2010}.
In the model it is assumed that the infectious disease induces partial
immunity. Denote by $S(t),\ I(t)$ and $R(t)$ the proportions of
susceptible population, infective population and recovered population
at time $t$, respectively. The partial immunity model is formulated
as \begin{subequations}\label{eq:SIRmodel} 
\begin{align}
S'(t) & =-\beta S(t)I(t),\label{eq:SIR1}\\
I^{\prime}(t) & =\beta S(t)I(t)+\beta\sigma R(t)I(t)-\gamma I(t),\label{eq:SIR2}\\
R^{\prime}(t) & =\gamma I(t)-\beta\sigma R(t)I(t).\label{eq:SIR3}
\end{align}
\end{subequations}The positive parameters $\beta$ and $\gamma$
are the transmission coefficient and the recovery rate, respectively.
The parameter $\sigma$ is the relative susceptibility of recovered
individuals, who have been infected at least once and have recovered
from the infection. We obtain the standard SIR epidemic model, if
$\sigma=0$, i.e., recovered individuals are completely protected
from the infection.

In this paper the partial immunity model (\ref{eq:SIRmodel}) is modified
as follows. When a recovered individual is exposed to the force of
infection, immunity is boosted with probability $1-\alpha$ so that
one obtains permanent immunity to the disease, while one contracts
the disease again with probability $\alpha$. The partial immunity
model (\ref{eq:SIRmodel}) is modified as \begin{subequations}\label{eq:SIRBmodel}
\begin{align}
S'(t) & =-\beta S(t)I(t),\label{eq:SIRB1}\\
I^{\prime}(t) & =\beta S(t)I(t)-\gamma I(t)+\beta\sigma\alpha I(t)R(t),\label{eq:SIRB2}\\
R^{\prime}(t) & =\gamma I(t)-\beta\sigma I(t)R(t),\label{eq:SIRB3}\\
B^{\prime}(t) & =\beta\sigma\left(1-\alpha\right)I(t)R(t)\label{eq:SIRB4}
\end{align}
\end{subequations}with the following initial conditions 
\begin{align*}
 & S(0)>0,\ I(0)>0,\ R(0)\geq0,\ B(0)\geq0,\\
 & S(0)+I\left(0\right)+R\left(0\right)+B\left(0\right)=1.
\end{align*}
Here $B(t)$ denotes the proportion of population with permanent immunity
at time $t$. We obtain the model (\ref{eq:SIRmodel}) by $\alpha=1$
and the SIR model by $\alpha=0$. Throughout the paper, we assume
the following two conditions 
\begin{align}
0 & <\sigma,\label{eq:ass_sigma}\\
0 & <\alpha<1.\label{eq:ass_alpha}
\end{align}

\section{One epidemic peak for $\mathcal{R}_{0}>1$}

We define the basic reproduction number by
\[
\mathcal{R}_{0}:=\frac{\beta}{\gamma}\left(S\left(0\right)+\alpha\sigma R\left(0\right)\right).
\]
The basic reproduction number is the expected number of secondary
cases produced by one infective individual in the expected \textit{one}
infectious period, $\frac{1}{\gamma}$ in the initial phase of epidemic.
Noting that both susceptible and recovered populations, which compose
the initial host population, have susceptibility to the disease, we
may call $\mathcal{R}_{0}$ the basic reproduction number, although
$\mathcal{R}_{0}$ is conventionally called the \textit{effective}
reproduction number \citep{Inaba:2017}.

From (\ref{eq:SIRB1}) and (\ref{eq:SIRB2}) one obtains the following
\begin{equation}
\frac{dI}{dS}=-\frac{\gamma}{\beta S}\left(\mathcal{R}(S,R)-1\right)\label{eq:IS}
\end{equation}
and 
\begin{equation}
\frac{dI(t)}{dt}=\gamma I(t)\left(\mathcal{R}(S(t),R(t))-1\right),\label{eq:It2}
\end{equation}
where 
\[
\mathcal{R}(S,R):=\frac{\beta}{\gamma}\left(S+\alpha\sigma R\right),\ S\geq0,\ R\geq0.
\]
Noting that $\mathcal{R}\left(S(0),R(0)\right)=\mathcal{R}_{0}$,
it is easy to see that 
\begin{align*}
\mathcal{R}_{0}>1\Leftrightarrow & I^{\prime}(0)>0,\\
\mathcal{R}_{0}=1\Leftrightarrow & I^{\prime}(0)=0,\\
\mathcal{R}_{0}<1\Leftrightarrow & I^{\prime}(0)<0,
\end{align*}
i.e., if $\mathcal{R}_{0}>1$ then the epidemic curve initially grows,
while if $\mathcal{R}_{0}<1$ then the epidemic curve initially decays.

First we show that $R(t)$ can be expressed in terms of $S(t)$. 
\begin{lem}
\label{lem:RbyS}It holds that 
\begin{equation}
R(t)=\frac{\gamma}{\sigma\beta}\left(1-\left(1-\frac{\sigma\beta}{\gamma}R\left(0\right)\right)\left(\frac{S\left(t\right)}{S\left(0\right)}\right)^{\sigma}\right),\ t\geq0.\label{eq:RbyS2}
\end{equation}
\end{lem}

\begin{proof}
Let us write $b$ for $\frac{\beta}{\gamma}$. Assume that $1-\sigma bR\left(0\right)\not=0$
holds. From the equations (\ref{eq:SIRB1}) and (\ref{eq:SIRB3})
we have 
\begin{equation}
\frac{dR}{dS}=-\frac{1-\sigma bR}{bS}.\label{eq:rs}
\end{equation}
Using the separation of variables, we obtain 
\begin{equation}
\left(\frac{S(t)}{S(0)}\right)^{\sigma}=\frac{1-\sigma bR\left(t\right)}{1-\sigma bR\left(0\right)},\label{eq:RbyS1}
\end{equation}
thus (\ref{eq:RbyS2}) follows. It is easy to see that the equality
in (\ref{eq:RbyS2}) also holds, if $1-\sigma bR\left(0\right)=0$.
\end{proof}
From Lemma \ref{lem:RbyS} we have 
\begin{equation}
R(t)=r(S(t)),\label{eq:r(S)}
\end{equation}
where

\begin{equation}
r(S):=\frac{\gamma}{\sigma\beta}\left(1-\left(1-\frac{\sigma\beta}{\gamma}R\left(0\right)\right)\left(\frac{S}{S(0)}\right)^{\sigma}\right),\ 0\leq S\leq S(0).\label{eq:R(S)-1}
\end{equation}
To analyze the epidemic curve, we study the function $\mathcal{R}(S,R)$
with $R=r\left(S\right)$. Let 
\[
\hat{\mathcal{R}}(S):=\mathcal{R}(S,r\left(S\right))
\]
We then compute the first and second derivatives of $\hat{\mathcal{R}}$:
\begin{align}
\hat{\mathcal{R}}^{\prime}\left(S\right) & =\frac{\beta}{\gamma}\left(1+\alpha\sigma r^{\prime}\left(S\right)\right).\label{eq:Rpri}\\
\hat{\mathcal{R}}^{\prime\prime}\left(S\right) & =\frac{\beta\alpha\sigma}{\gamma}r^{\prime\prime}\left(S\right).\label{eq:Rpri2}
\end{align}
From the equation (\ref{eq:rs}) in Lemma \ref{lem:RbyS}, it is easy
to obtain the following result.
\begin{lem}
\label{lem:mono} One has
\begin{align}
r^{\prime}(S) & =-\frac{\gamma}{\beta}\left(1-\sigma bR(0)\right)\frac{S^{\sigma-1}}{S(0)^{\sigma}},\label{eq:rrpri}\\
r^{\prime\prime}(S) & =\left(\sigma-1\right)\frac{1}{S}r^{\prime}(S).\label{eq:rrpri2}
\end{align}
\end{lem}

Note that $r$ is a monotone function, thus $\hat{\mathcal{R}}$ has
at most one extremum. 

We now show the standard epidemic case if $\mathcal{R}_{0}>1$ holds. 
\begin{prop}
\label{prop:hatRgraph1}Let us assume that $\mathcal{R}_{0}>1$ holds.
Then 
\begin{equation}
\hat{\mathcal{R}}(0)=\alpha<1<\hat{\mathcal{R}}(S(0))=\mathcal{R}_{0}.\label{eq:edge}
\end{equation}
holds and there exists a unique root of 
\[
\hat{\mathcal{R}}(S)=1,\ 0<S<S(0).
\]
\end{prop}

\begin{proof}
It is easy to see that (\ref{eq:edge}) holds. First assume that $r^{\prime}(S)\geq0$
for $0\leq S\leq S(0)$. Then, from (\ref{eq:Rpri}), one can see
that $\hat{\mathcal{R}}$ is an increasing function. Thus we obtain
the conclusion. Next assume that $r^{\prime}(S)<0$ for $0\leq S\leq S(0)$.
By Lemma \ref{lem:mono}, one sees that $\hat{\mathcal{R}}$ has at
most one extremum for $0\leq S\leq S(0)$. Therefore, from (\ref{eq:edge}),
we obtain the conclusion.
\end{proof}
Then, from Proposition \ref{prop:hatRgraph1} and Lemma \ref{lem:limitexists}
in Appendix \ref{sec:AppA}, we obtain the following result.
\begin{thm}
\label{thm:outbreak_R0>1}Let us assume that $\mathcal{R}_{0}>1$
holds. Then there is a $t_{p}>0$ such that $I$ is monotonically
increasing for $t\in\left(0,t_{p}\right)$ and monotonically decreasing
for $t>t_{p}$. It holds $\lim_{t\to\infty}I(t)=0$. 
\end{thm}

\section{Delayed epidemic for $\mathcal{R}_{0}\leq1$}

In the standard SIR model, when $\mathcal{R}_{0}\leq1$ holds, then
the epidemic curve monotonically decreases and infective population
tends to $0$ eventually as time goes to infinity. The situation changes
in the model with boosting immunity (\ref{eq:SIRBmodel}), due to
the susceptibility of the recovered individuals. In particular, if
$\sigma>1$ then there is a possible delayed outbreak as the recovered
population increases which will induce the epidemic later even if
$\mathcal{R}_{0}\leq1$. The basic reproduction number, which characterizes
the initial dynamics, is not a sufficient criterion to determine the
outbreak due to the recovered population.

First let us consider a simple case that $\sigma\leq1$ holds. We
have the standard scenario: if $\mathcal{R}_{0}\leq1$ then the epidemic
does not occur. Subsequently we study the disease transmission dynamics
when $\sigma>1$. We show that enhancement of susceptibility after
the infection can induce an epidemic later.

\subsection{$\sigma\leq1$}

We show that the infective population is monotonically decreasing
for $t\geq0$, similar to the SIR model, when $\mathcal{R}_{0}\leq1$. 
\begin{prop}
\label{prop:hatRgraph1-1}Let us assume that $\mathcal{R}_{0}\leq1$
and $\sigma\leq1$ holds. Then 
\[
\mathcal{\hat{R}}(S)\leq1,\ 0\leq S\leq S(0).
\]
\end{prop}

\begin{proof}
Note that 

\begin{equation}
\hat{\mathcal{R}}(0)=\alpha<1,\ \hat{\mathcal{R}}(S(0))=\mathcal{R}_{0}\leq1\label{eq:edge-1}
\end{equation}
holds. Assume that $r^{\prime}(S)\geq0$ for $0<S<S(0)$. Then $\hat{\mathcal{R}}$
is an increasing function, thus we obtain the conclusion. Next assume
that $r^{\prime}(S)<0$ for $0<S<S(0)$. In this case one sees that
\[
\lim_{S\downarrow0}r^{\prime}(S)=-\infty\implies\lim_{S\downarrow0}\hat{\mathcal{R}}^{\prime}(S)=-\infty.
\]
By Lemma \ref{lem:mono}, one sees that $\hat{\mathcal{R}}$ has at
most one minimum for $0\leq S\leq S(0)$. Therefore we obtain the
conclusion.
\end{proof}
Then, from Propositions \ref{prop:hatRgraph1-1} and Lemma \ref{lem:limitexists}
in Appendix \ref{sec:AppA}, we obtain the following result.
\begin{thm}
Let us assume that $\mathcal{R}_{0}\leq1$ and $\sigma\leq1$ hold.
Then $I$ is monotonically decreasing for $t\geq0$. It holds $\lim_{t\to\infty}I(t)=0$. 
\end{thm}

\begin{figure*}
\begin{centering}
\subfloat[Graph of $\hat{\mathcal{R}}(S)-1$]{\includegraphics[scale=0.7]{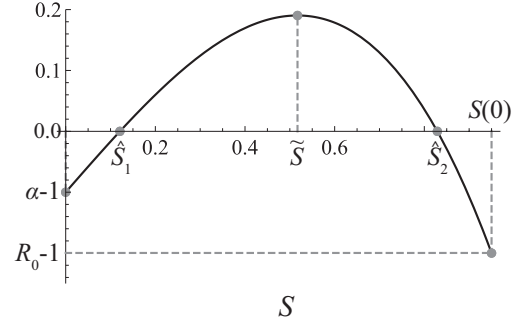}}
\par\end{centering}
\begin{centering}
\subfloat[Phase portrait of the solution $(I((t),S(t))$ in the $(I,S)$ plane]{\includegraphics[scale=0.7]{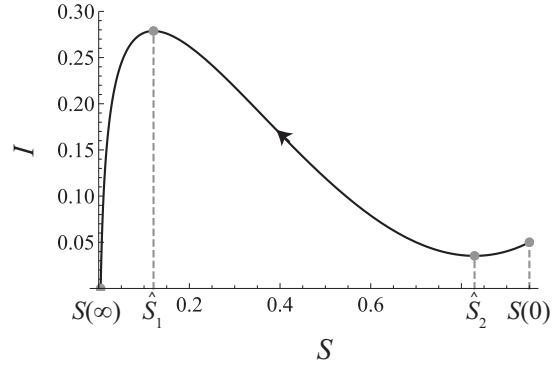}}
\par\end{centering}
\centering{}\caption{\label{fig:funcf}The graph of $\hat{\mathcal{R}}(S)-1$ for $0\leq S\leq S(0)$
is shown in (A). Here $\mathcal{R}_{0}\leq1$ and $\sigma>1$. The
parameters are chosen such that (\ref{eq:fdiff0_2}) and $\hat{\mathcal{R}}(\tilde{S})>1$
hold. The parametric curve $t\to(I(t),S(t))$ in the $(I,S)$ plane
is plotted in (B). It is shown that $I$ has a local minima at $S=\hat{S}_{2}$
and local maxima at $S=\hat{S}_{1}$.}
\end{figure*}

\subsection{$\sigma>1$}

In this subsection we consider the case that 
\begin{equation}
\mathcal{R}_{0}\leq1,\ \sigma>1\label{eq:ass}
\end{equation}
hold. We show the following results for the graph of $\hat{\mathcal{R}}$. 
\begin{prop}
\label{prop:f_sigma>1}Let us assume that $\mathcal{R}_{0}\leq1$
and $\sigma>1$ hold. 
\begin{enumerate}
\item If 
\begin{equation}
\frac{\beta}{\gamma}\left(S(0)+\alpha\sigma^{2}R(0)\right)\geq\alpha\sigma\label{eq:fdiff0_1}
\end{equation}
then $\hat{\mathcal{R}}(S)\leq1$ for $0\leq S\leq S(0)$. 
\item If 
\begin{equation}
\frac{\beta}{\gamma}\left(S(0)+\alpha\sigma^{2}R(0)\right)<\alpha\sigma\label{eq:fdiff0_2}
\end{equation}
then there is a unique maxima for $0<S<S(0)$ at $S=\tilde{S}$, where
\begin{equation}
\tilde{S}:=\left(\frac{\frac{\beta}{\gamma}S(0)}{\sigma\alpha\left(1-\frac{\sigma\beta}{\gamma}R(0)\right)}\right)^{\frac{1}{\sigma-1}}S(0)<S(0).\label{eq:Stilde}
\end{equation}
Then
\begin{enumerate}
\item If $\hat{\mathcal{R}}(\tilde{S})>1$ then there are two roots for
$\hat{\mathcal{R}}(S)=1$ for $0<S<S(0)$. Denote the roots by $\hat{S}_{1}$
and $\hat{S}_{2}$ such that 
\[
0<\hat{S}_{1}<\tilde{S}<\hat{S}_{2}<S(0),
\]
then 
\[
\hat{\mathcal{R}}(S)\begin{cases}
<1, & \hat{S}_{2}<S\leq S(0),\\
>1, & \hat{S}_{1}<S<\hat{S}_{2},\\
<1, & 0<S<\hat{S}_{1}.
\end{cases}
\]
\item If $\text{\ensuremath{\hat{\mathcal{R}}}}(\tilde{S})\leq1$ then $\hat{\mathcal{R}}(S)\leq1$
for $0\leq S\leq S(0)$.
\end{enumerate}
\end{enumerate}
\end{prop}

\begin{proof}
For $\sigma>1$ one sees that 
\[
\lim_{S\downarrow0}r^{\prime}(S)=0\implies\lim_{S\downarrow0}\hat{\mathcal{R}}^{\prime}(S)=\frac{\beta}{\gamma}>0.
\]
From the monotonicity of $\hat{\mathcal{R}}^{\prime}$, if $\lim_{S\uparrow S(0)}\hat{\mathcal{R}}^{\prime}(S)>0$
then $\hat{\mathcal{R}}(S)\leq1$ for $0\leq S\leq S(0)$ follows.
Computing 

\begin{align*}
\lim_{S\uparrow S(0)}\hat{\mathcal{R}}^{\prime}(S) & =\frac{\beta}{\gamma}\left(1+\alpha\sigma r^{\prime}(S(0))\right)\\
 & =\frac{1}{S(0)}\left(\frac{\beta}{\gamma}\left(S(0)+\alpha\sigma^{2}R(0)\right)-\alpha\sigma\right),
\end{align*}
one can see that (\ref{eq:fdiff0_1}) is equivalent to that $\lim_{S\uparrow S(0)}\hat{\mathcal{R}}^{\prime}(S)>0$
holds. 

Next assume that (\ref{eq:fdiff0_2}) holds. Then 
\[
\lim_{S\uparrow S(0)}\hat{\mathcal{R}}^{\prime}(S)<0<\lim_{S\downarrow0}\hat{\mathcal{R}}^{\prime}(S).
\]
From the monotonicity of $\hat{\mathcal{R}}^{\prime},$ there is a
unique maxima for $0<S<S(0)$. Solving $\hat{\mathcal{R}}^{\prime}(S)=0$,
we obtain $\tilde{S}$ given as in (\ref{eq:Stilde}). It is now straightforward
to obtain the statements (a) and (b).
\end{proof}
In Figure \ref{fig:funcf}, we plot the graph of the function $\hat{\mathcal{R}}(S)-1$
for $0\leq S\leq S(0)$, where $\mathcal{R}_{0}\leq1$ and $\sigma>1$.
Parameters are fixed so that (\ref{eq:fdiff0_2}) and $\hat{\mathcal{R}}(\tilde{S})>1$
hold. 

From Proposition \ref{prop:f_sigma>1} and Lemma \ref{lem:limitexists}
in Appendix \ref{sec:AppA}, we first obtain the result for the extinction
of the disease.
\begin{thm}
Let us assume that $\mathcal{R}_{0}\leq1$ and $\sigma>1$ holds.
If either that 
\begin{enumerate}
\item (\ref{eq:fdiff0_1}) holds, or 
\item (\ref{eq:fdiff0_2}) and $\hat{\mathcal{R}}(\tilde{S})\leq1$ hold, 
\end{enumerate}
then $I$ is monotonically decreasing for $t\geq0$. It follows that
$\lim_{t\to\infty}I(t)=0$. 
\end{thm}

Now it is assumed that (\ref{eq:ass}) holds. If 
\[
\hat{\mathcal{R}}(\tilde{S})>1
\]
holds, where $\tilde{S}$ is a root of 
\[
\hat{\mathcal{R}}^{\prime}(S)=0
\]
and the existence is ensured by the condition (\ref{eq:fdiff0_2}),
then $I(t)$ may attain a minimum and a maxima (see Figure \ref{fig:funcf}).
This implies that even if $\mathcal{R}_{0}\le1$, the epidemic curve
may grow for a certain time interval, which we call \textit{delayed
epidemic.}

To determine if the delayed epidemic indeed occurs, we evaluate the
minimum of $I(t)$ using the following expression for $I$ derived
in Proposition \ref{proposition:IB_Exp} in Appendix \ref{sec:AppA}
\begin{equation}
I(t)=I(0)+p(S(t))-\alpha q(R(t)),\ t\geq0,\label{eq:I_sol}
\end{equation}
where

\begin{align*}
p(S) & :=\left(S(0)-S\right)+\frac{\gamma}{\beta}\ln\left(\frac{S}{S(0)}\right),\\
q(R) & :=\frac{\gamma}{\sigma\beta}\ln\left(\frac{1-\frac{\sigma\beta}{\gamma}R}{1-\frac{\sigma\beta}{\gamma}R(0)}\right)+\left(R-R(0)\right).
\end{align*}
Substituting (\ref{eq:r(S)}) into (\ref{eq:I_sol}), $I(t)$ can
be expressed in terms of $S(t)$ as follows 
\[
I(t)=I(0)+p(S(t))-\alpha q(r(S(t)),\ t\geq0.
\]
See also Figure \ref{fig:funcf} (B) for the phase portrait in the
$(I,S)$-plane. 
\begin{thm}
\label{theorem:delay}Let us assume that $\mathcal{R}_{0}\leq1$ and
$\sigma>1$ holds. Furthermore, assume that (\ref{eq:fdiff0_2}) and
\[
\hat{\mathcal{R}}(\tilde{S})>1
\]
hold. 
\begin{enumerate}
\item If $I(0)+p(\hat{S}_{2})-\alpha q(r(\hat{S}_{2}))\leq0$, then $I$
is monotonically decreasing for $t\geq0$. 
\item If $I(0)+p(\hat{S}_{2})-\alpha q(r(\hat{S}_{2}))>0$, then there is
an interval $\left[t_{1},t_{2}\right]$ such that $I$ increases for
$t_{1}\leq t\leq t_{2}$ and decreases for $0\le t\leq t_{1}$ and
$t_{2}\leq t$. 
\end{enumerate}
It follows that $\lim_{t\to\infty}I(t)=0$.

\end{thm}

\begin{proof}
One sees that $I$ has a local maxima and minima with respect to $t$
and $S$, where $\hat{\mathcal{R}}(S)=1$ holds (see \ref{eq:IS}
and \ref{eq:It2}). $I$ has a local minima at $S=\hat{S}_{2}\in\left(\tilde{S},S(0)\right)$
and $I$ is increasing for $\hat{S}_{2}\leq S\leq S(0)$ (see Figure
\ref{fig:funcf} (B)). Noting that $I(t)>0$ for $t\geq0$ and that
$S$ is a decreasing function with respect to $t$, $I(0)+p(\hat{S}_{2})-\alpha q(r(\hat{S}_{2}))\leq0$
implies that $I$ is monotonically decreasing for $t\geq0$. On the
other hand, if $I(0)+p(\hat{S}_{2})-\alpha q(r(\hat{S}_{2}))>0$ then
$I$ is monotonically increasing for $S<\hat{S}_{1}$, decreasing
for $\hat{S}_{1}<S<\hat{S}_{2}$ and then increasing for $\hat{S}_{2}<S$.
There exist $t_{1}$ and $t_{2}$ such that $S(t_{1})=\hat{S}_{2}$
and $S(t_{2})=\hat{S}_{1}$. Thus we obtain the conclusion. From Lemma
\ref{lem:limitexists} in Appendix \ref{sec:AppA} it follows that
$\lim_{t\to\infty}I(t)=0$. 
\end{proof}
Thus the model has three different transmission dynamics: no epidemic,
normal epidemic and delayed epidemic as illustrated in Figure \ref{fig:epicurve}.
Figure \ref{fig:paramreg} shows parameter regions for the three different
disease transmission dynamics. The region for the delayed epidemic
become larger with respect to the initial condition of $I$ and the
susceptibility $\sigma$. Since the initial condition of $I$ is involved
in the condition of Theorem \ref{theorem:delay}, the initial condition
qualitatively changes the epidemic curve, see Figure \ref{fig:delay_init}:
delayed epidemic is induced by a large initial condition. 
\begin{figure*}
\begin{centering}
\includegraphics[scale=0.45]{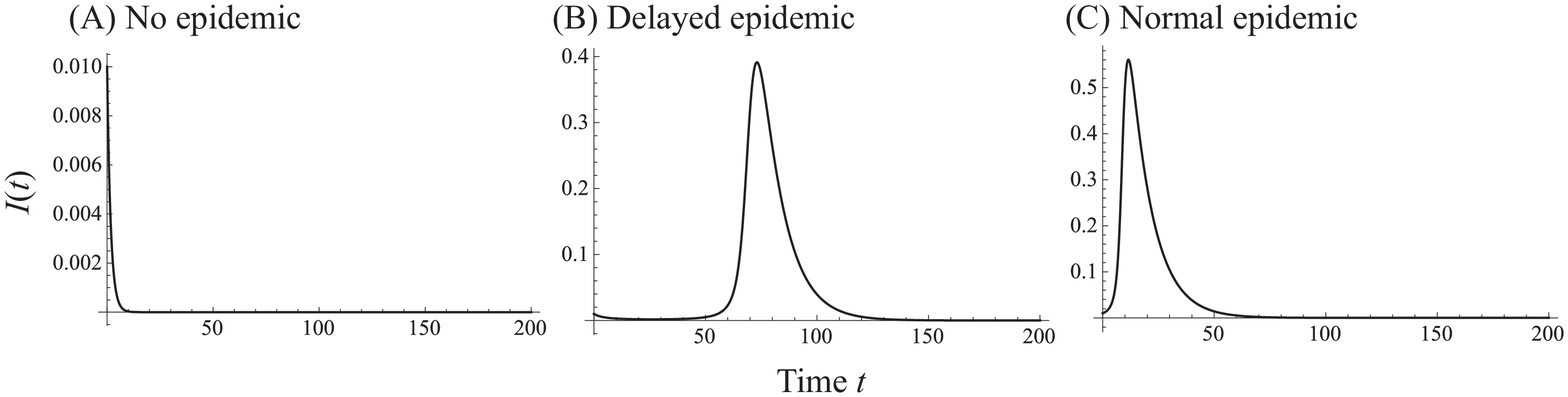}
\par\end{centering}
\centering{}\caption{\label{fig:epicurve}Examples of three types of epidemic curve. (A)
shows no epidemic case, (B) shows delayed epidemic, and (C) shows
normal epidemic, respectively. Parameters were set as $\mathcal{R}_{0}=0.4$
for (A), $0.8$ for (B), and $1.2$ for (C), other parameter values
are identical between (A), (B) and (C): $\alpha=0.9,\:\sigma=5,\ S(0)=0.99,\:I(0)=0.01$,
and $R(0)=B(0)=0$.}
\end{figure*}

\begin{figure*}
\begin{centering}
\includegraphics[scale=0.27]{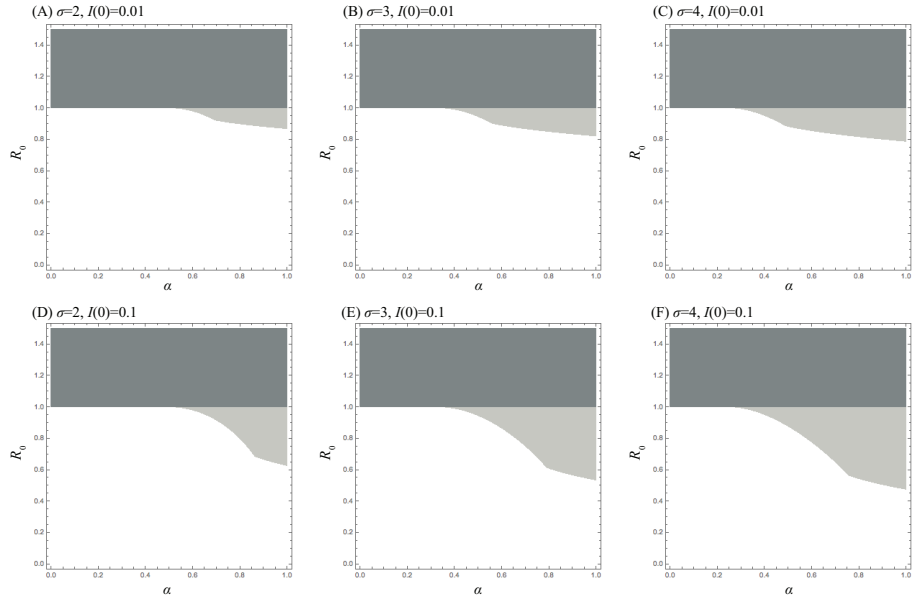}
\par\end{centering}
\centering{}\caption{\label{fig:paramreg}The dependency of epidemic type on $\mathcal{R}_{0}$
and $\alpha$ for several $\sigma$ and $I\left(0\right)$. White
are denotes ``no epidemic'', light gray area denotes ``delayed
epidemic'', and gray area denotes ``normal epidemic'', respectively.}
\end{figure*}

\begin{figure*}
\begin{centering}
\includegraphics[scale=0.55]{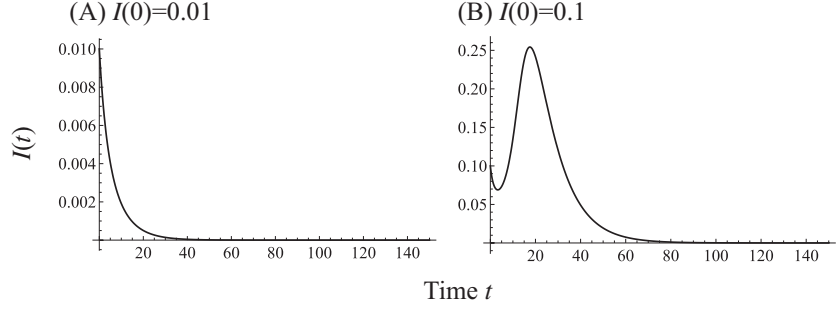}
\par\end{centering}
\centering{}\caption{\label{fig:delay_init}The initial condition for $I$ changes the
disease transmission dynamics. Here parameters are chosen as $\beta=0.8,\ \gamma=1,\ \alpha=0.9,\ \sigma=3$.
$R(0)=B(0)=0$. (A) shows no epidemic for $I(0)=0.1$ while (B) shows
delayed epidemic for $I(0)=0.01$.}
\end{figure*}

Consider a special case that $R\left(0\right)\to0$. The basic reproduction
number is given as 
\[
\mathcal{R}_{0}=\frac{\beta}{\gamma}S(0).
\]
The conditions (\ref{eq:fdiff0_2}) becomes 
\[
\mathcal{R}_{0}<\sigma\alpha
\]
and $\tilde{S}=\left(\frac{\mathcal{R}_{0}}{\sigma\alpha}\right)^{\frac{1}{\sigma-1}}S(0).$
If $\mathcal{\hat{R}}(\tilde{S})>1$ holds, then the delayed epidemic
may occur. 

\section{Final epidemic size }

Let 
\[
\left(S(\infty),I(\infty),R(\infty),B(\infty)\right)=\lim_{t\to\infty}\left(S(t),I(t),R(t),B(t)\right).
\]
It follows that $I(\infty)=0$. From the relations (\ref{eq:IbySR}),
(\ref{eq:RbyS2}) and (\ref{eq:BbySR}), one sees that $(S(\infty),R(\infty),B(\infty))$
satisfy the following equations

\begin{align}
0 & =I(0)+p(S\left(\infty\right))-\alpha q(R\left(\infty\right)),\label{eq:Sf}\\
R(\infty) & =r(S(\infty)),\label{eq:Rf}\\
B(\infty) & =B(0)-\left(1-\alpha\right)q(R(\infty)).\label{eq:Bf}
\end{align}
The final epidemic size is given by $R(\infty)+B(\infty)$, the number
of individuals who infected at least once. From (\ref{eq:Sf}) and
(\ref{eq:Rf}) we get the following equation
\begin{equation}
0=I(0)+p(S\left(\infty\right))-\alpha q(r\left(S(\infty)\right)).\label{eq:Sinf}
\end{equation}
In Figure \ref{fig:FinalRB}, we plot $R(\infty)+B(\infty)=1-S(\infty),\ R(\infty)$
and $B(\infty)$ with respect to $\mathcal{R}_{0}$.

Numerically we observe that $R(\infty)$ is not monotone with respect
to $\mathcal{R}_{0}$. Small $\mathcal{R}_{0}$ allows the increase
of $R(\infty)$, on the other hand, does not contribute to the increase
of $B(t)$, the outbreak ends before the transition from $R(t)$ to
$B(t)$ via $I(t)$ occurs among most $R(t)$. Increase of $\mathcal{R}_{0}$
contributes the transition from $R(t)$ to $B(t)$, consequently,
$R(\infty)$ decreases. Despite of non-monotnic relation of $R(\infty)$
with respect to $\mathcal{R}_{0}$, $R(\infty)+B(\infty)$ is likely
to increase monotonically with the increase of $\mathcal{R}_{0}$
as shown in Figure \ref{fig:FinalRB}.

When $\alpha=0$ we obtain the standard SIR setting. Letting $I(0)\to0$
and $S(0)\to1$, the basic reproduction number is given as $\mathcal{R}_{0}=b$.
In this case, from (\ref{eq:Sf}), we obtain the well known final
size relation 
\[
0=\left(1-S\left(\infty\right)\right)+\frac{1}{\mathcal{R}_{0}}\ln\left(S\left(\infty\right)\right),
\]
see e.g. \citep{Diekmann:2012,Inaba:2017}.

\begin{figure}
\begin{centering}
\includegraphics[scale=0.55]{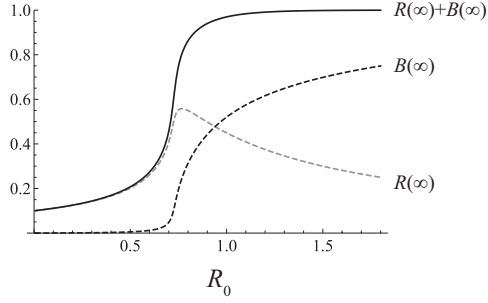}
\par\end{centering}
\centering{}\caption{\label{fig:FinalRB}Final epidemic size with respect to $\mathcal{R}_{0}$.
Initial conditions are fixed as $S(0)=0.9,\ I(0)=0.1,\ R(0)=B(0)=0$.
Parameters are fixed as $\sigma=2,\ \alpha=0.8$.}
\end{figure}

\section{Discussion}

In this paper we study a disease transmission dynamics model incorporating
natural-boosting immunity. Our modelling approach describing boosting
immunity covers not only the standard transmission dynamics but also
an interesting dynamics, delayed epidemic. Delayed epidemic shows
negative slope at the initial phase of epidemic, thus, the estimation
of $\mathcal{R}_{0}$ using the initial slope of epidemic is difficult
to capture the actual epidemic coming later. We derive the condition
for a delayed epidemic through deriving the analytic transient solution
of $I(t)$.

Delayed epidemic, which is illustrated in Figures \ref{fig:epicurve}
and \ref{fig:delay_init}, occurs due to the enhancement of susceptibility
of the recovered population (i.e., $\sigma>1$). For example, antibody
dependent enhancement can enhance the viral replication within the
host body, consequently, the host susceptibility can be enhanced at
the time of reinfection \citep{Whitehead:2007}. In Theorem \ref{theorem:delay}
we formulate a condition for the delayed epidemic. One of the necessary
condition for the delayed epidemic is (\ref{eq:fdiff0_2}) in Proposition
\ref{prop:f_sigma>1}. The condition (\ref{eq:fdiff0_2}) is necessary
for increasing of the epidemic curve and is related to increasing
of the effective susceptible population, which is defined as

\[
L(t):=S(t)+\alpha\sigma R(t).
\]
Since it holds that 
\begin{align*}
L^{\prime}(0) & =S^{\prime}(0)+\alpha\sigma R^{\prime}(0)\\
 & =\gamma I(0)\left(-bS(0)+\alpha\sigma\left(1-b\sigma R(0)\right)\right),
\end{align*}
one can see that 
\begin{align*}
L^{\prime}(0)>0 & \Leftrightarrow bS(0)<\sigma\alpha\left(1-b\sigma R(0)\right),
\end{align*}
where $b=\frac{\beta}{\gamma}$. Therefore, increasing of the effective
susceptible population at the initial time is necessary for the delayed
epidemic and may induce the delayed epidemic even if $\mathcal{R}_{0}\leq1$
holds.

We remark that $\mathcal{R}_{0}$ cannot measure the outbreak potential
of ``delayed epidemic''. In principe, $\mathcal{R}_{0}$ is derived
based on the linearized system at the the initial disease transmission
dynamics. The linearized system at the the initial phase does not
provide enough information to predict the delayed epidemic. Similar
phenomena can be observed in epidemic models that show backward bifurcation
of the endemic equilibrium \citep{Arino:2003,Brauer:2004,Greenhalgh:2000,Inaba:2017,Kribs:2000}.
In those studies it is shown that there is a stable endemic equilibrium
even if the basic reproduction number is less than unity. Differently
from those models, the short-term disease transmission dynamics model
have many equilibria, which are associated to the zero eigenvalue.
Our study illustrates that, in the short-term disease transmission
dynamics, the outbreak potential shall be carefully examined, using
the transient solution.

We also observed the reinfection threshold like behavior \citep{Gomes:2004,Katriel:2010,Inaba:2016}
(see Figure \ref{fig:FinalRB}). In the extreme case that the initial
population is composed of only $I(0)$ and $R(0)$ ($S(0)=B(0)=0$),
$\mathcal{R}_{0}=\frac{\beta}{\gamma}\alpha\sigma R\left(0\right)=1$
is shown to be the threshold for the outbreak, which amounts to the
concept of the reinfection threshold. Differently from the models
studied in \citep{Gomes:2004,Katriel:2010,Inaba:2016}, our model
has the full protection compartment $B$. We here found that reinfection
threshold is not a sufficient criterion for the outbreak if the initial
population is composed of $S(0),I(0),R(0)$ and $B(0)$ (gray area
shown in Figure \ref{fig:FinalRB}).

$\mathcal{R}_{0}$ can be estimated from the final epidemic size.
It should be noted that $\mathcal{R}_{0}$ can be overestimated if
the model neglects the boosting immunity. Figure \ref{fig:over} shows
the estimated $\mathcal{R}_{0}$ using a fixed final epidemic size
$=0.5$ with varied $\alpha$ and $\sigma$, our model is equivalent
with a standard SIR model when $\alpha=0$ or $\sigma=0$. If boosting
and waning immunity are introduced, $\alpha>0$ or $\sigma>0$, the
estimated $\mathcal{R}_{0}$ is always lower than it using the standard
epidemic model, $\alpha=0$ or $\sigma=0$. To estimate the precise
$\mathcal{R}_{0}$ from the final epidemic size, the appropriate modelling
with respect to boosting and waning immunity is required.

\begin{figure}
\includegraphics[scale=0.45]{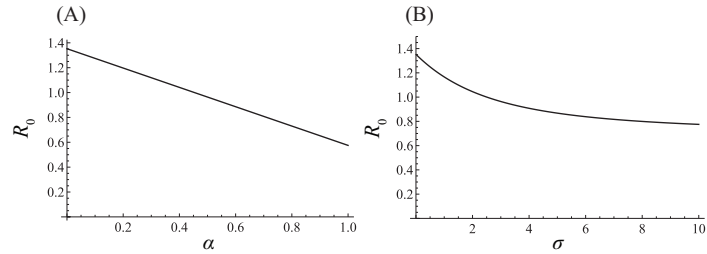}

\caption{\label{fig:over}$\mathcal{R}_{0}$ derived from a given final epidemic
size, $R(\infty)+B(\infty)=0.5$ with varied $\alpha$ and $\sigma$.
When $\alpha=0$ or $\sigma=0$ the boosting and waning immunity do
not occur (the standard SIR model). The initial conditions are fixed
as $S(0)=0.99,\ I(0)=0.01,\ R(0)=B(0)=0$. Parameters are fixed as
$\sigma=3$ for (A) and $\alpha=0.5$ for (B).}
\end{figure}

The time series data of the reported $I(t)$ is used to estimate the
epidemiological parameters. However, the reported $I(t)$ can be biased
by reporting biases and asymptomatic cases. Serological surveillance
can collect the data which is less likely to suffer from such biases.
Our analytical results allows the real-time estimation of $I(t)$
using the field data obtained by sero-surveillance. Since $I(t)$
can be implicitly determined from $R(t)+B(t)$ in our model. $1-\ensuremath{(R(t)+B(t))}=S(t)+I(t)$
and $I(t)$ is a function of $S(t)$, then $I(t)$ can be derived
from $R(t)+B(t)$. If $R(t)+B(t)$ is collected by serological study,
$I(t)$ can be estimated.

Our mathematical model describing natural-boosting immunity has a
limitation; we described step-wise level of boosting immunity, i.e.,
$R$ has susceptibility $\sigma$ to the infectious disease while
$B$ has a complete protection against reinfection. This setting is
suitable for the infectious diseases such that the multiple infections
can establish drastic increase of the immunity level. On the other
hand, to describe gradual change of the immunity level resulted from
boosting and waning immunity, the several classes of $R$ with varied
immune protection level are required.

\subsection*{Acknowledgement}

The first author was supported by JSPS Grant-in-Aid for Young Scientists
(B) 16K20976 of Japan Society for the Promotion of Science. The second
author was supported by PRESTO, Japan Science and Technology Agency,
grant number JPMJPR15E1, and JSPS Grant-in-Aid for Young Scientists
(B) 15K19217 of Japan Society for the Promotion of Science.

\appendix

\section{Disease transmission dynamics\label{sec:AppA}}

For simplicity, we write $b$ for $\frac{\beta}{\gamma}$.
\begin{lem}
\label{lem:limitexists}There exist $\lim_{t\to\infty}S(t),\ \lim_{t\to\infty}I(t),\ \lim_{t\to\infty}R(t)$
and $\lim_{t\to\infty}B(t)$. It holds that 
\begin{equation}
\lim_{t\to\infty}I(t)=0.\label{eq:limI0}
\end{equation}
\end{lem}

\begin{proof}
One easily sees that $S,\ R$ and $B$ are respectively monotone bounded
functions. Specifically $S$ is a monotone decreasing function, while
$B$ is a monotone increasing function. Therefore, $S,\ R$ and $B$
tend to some constants. Since $S(t)+I(t)+R(t)+B(t)=1$ holds for $t\geq0$,
$\lim_{t\to\infty}I(t)$ also exists. We now claim that (\ref{eq:limI0})
holds. From (\ref{eq:SIRB1}), (\ref{eq:SIRB2}) and (\ref{eq:SIRB3})
one has 
\[
S^{\prime}(t)+I^{\prime}(t)+\alpha R^{\prime}(t)=-\gamma\left(1-\alpha\right)I(t).
\]
Suppose that $\lim_{t\to\infty}I(t)>0$. Integrating the above equation,
we derive a contradiction. Hence (\ref{eq:limI0}) holds. 
\end{proof}
We now introduce the following lemma.
\begin{lem}
\label{lemma:intRS}One has 
\begin{align}
\int_{S(0)}^{S(t)}\frac{\sigma r(S)}{S}dS= & \frac{1}{\sigma b}\ln\left(\frac{1-\sigma bR(t)}{1-\sigma bR(0)}\right)+\left(R(t)-R(0)\right),\ t\geq0.\label{eq:intRS}
\end{align}
\end{lem}

\begin{proof}
We compute 
\begin{align*}
\int\frac{\sigma r(S)}{S}dS= & \sigma\int\frac{1}{S}\left[\frac{1}{\sigma b}\left(1-\left(1-\sigma bR(0)\right)\left(\frac{S}{S(0)}\right)^{\sigma}\right)\right]dS\\
= & \frac{1}{b}\left[\int\frac{1}{S}dS-\left(1-\sigma bR(0)\right)\int\left(\frac{S}{S(0)}\right)^{\sigma}\frac{1}{S}dS\right].
\end{align*}
First we have 
\[
\int_{S(0)}^{S(t)}\frac{1}{S}dS=\ln\left(\frac{S(t)}{S(0)}\right).
\]
From (\ref{eq:RbyS1}) and (\ref{eq:rs}) in the proof of Lemma \ref{lem:RbyS}
we get 
\begin{align*}
\int\left(\frac{S}{S(0)}\right)^{\sigma}\frac{1}{S}dS & =-b\int\left(\frac{1-\sigma bR}{1-\sigma bR(0)}\right)\frac{1}{1-\sigma bR}dR\\
 & =-b\frac{1}{1-\sigma bR(0)}\int dR.
\end{align*}
Therefore, we get 
\[
\int_{S(0)}^{S(t)}\left(\frac{S}{S(0)}\right)^{\sigma}\frac{1}{S}dS=-b\frac{1}{1-\sigma bR(0)}\left(R(t)-R(0)\right).
\]
Then 
\[
\int_{S(0)}^{S(t)}\frac{\sigma r(S)}{S}dS=\frac{1}{b}\ln\left(\frac{S(t)}{S(0)}\right)+\left(R(t)-R(0)\right).
\]
From (\ref{eq:RbyS1}) in the proof of Lemma \ref{lem:RbyS}, \textbf{ 
\[
\ln\left(\frac{S(t)}{S(0)}\right)=\frac{1}{\sigma}\ln\left(\frac{1-\sigma bR(t)}{1-\sigma bR(0)}\right).
\]
}Finally we thus obtain (\ref{eq:intRS}). 
\end{proof}
Then we show explicit expressions for $I$ and $B$ in terms of $S$
and $R$.
\begin{prop}
\label{proposition:IB_Exp}One has that 
\begin{align}
I(t)= & I(0)+\left(S(0)-S(t)\right)+\frac{1}{b}\ln\left(\frac{S(t)}{S(0)}\right)\nonumber \\
 & -\alpha\left\{ \frac{1}{\sigma b}\ln\left(\frac{1-\sigma bR(t)}{1-\sigma bR(0)}\right)+\left(R(t)-R(0)\right)\right\} ,\label{eq:IbySR}\\
B(t)= & B(0)-\left(1-\alpha\right)\left\{ \frac{1}{\sigma b}\ln\left(\frac{1-\sigma bR(t)}{1-\sigma bR(0)}\right)+\left(R(t)-R(0)\right)\right\} .\label{eq:BbySR}
\end{align}
for $t\geq0$. 
\end{prop}

\begin{proof}
From (\ref{eq:SIRB1}) and (\ref{eq:SIRB2}) we have 
\begin{equation}
\frac{dI}{dS}=-1+\frac{1}{bS}-\alpha\sigma\frac{R}{S}.\label{eq:dIdS}
\end{equation}
We use the separation of variables to obtain (\ref{eq:IbySR}). Using
(\ref{eq:intRS}) in Lemma \ref{lemma:intRS} one obtains 
\[
\int_{S(0)}^{S(t)}\left(\frac{1}{bS}-\alpha\sigma\frac{r(S)}{S}\right)dS=\frac{1}{b}\ln\left(\frac{S(t)}{S(0)}\right)-\alpha\left\{ \frac{1}{\sigma b}\ln\left(\frac{1-\sigma bR(t)}{1-\sigma bR(0)}\right)+\left(R(t)-R(0)\right)\right\} .
\]
Therefore we get (\ref{eq:IbySR}). Next from (\ref{eq:SIRB1}) and
(\ref{eq:SIRB2}) we have 
\begin{equation}
\frac{dB}{dS}=-\left(1-\alpha\right)\sigma\frac{R}{S}\label{eq:dIdS-1}
\end{equation}
Using (\ref{eq:intRS}) in Lemma \ref{lemma:intRS}, one obtains (\ref{eq:BbySR}). 
\end{proof}

\end{document}